\documentclass[letterpaper,twocolumn,10pt]{article}
\usepackage{usenix}

\usepackage{tikz}
\usepackage{amsmath}

\usepackage{filecontents}

\usepackage{pgfplots}
\usepackage{comment}
\usepackage[linewidth=0.5pt]{mdframed}
\usepackage{graphicx}
\usepackage{framed}
\usepackage{xspace}
\usepackage{xcolor}
\usepackage{tabularx}
\usepackage{multicol}
\usepackage{multirow}
\usepackage{float}
\usepackage{subcaption}
\usepackage{pifont}
\usepackage{pdfpages}
\usepackage{cite}
\usepackage{hyperref}
\usepackage{bm}
\usepackage{tabularx}
\usepackage{multirow}
\usepackage{booktabs}
\usepackage{amsmath}

\hypersetup{colorlinks=true,citecolor=blue,filecolor=blue,linkcolor=blue,urlcolor=blue}

\usepackage{colortbl}
\usepackage{enumitem}
\usepackage[override]{cmtt}
\usepackage{graphicx,color}
\usepackage{boxedminipage}
\usepackage{url}
\usepackage{caption}
\usepackage{amsmath,amsthm,amstext,amssymb,amsfonts,latexsym}
\usepackage{algorithm}
\usepackage[noend]{algpseudocode}
\usepackage{algorithmicx}
\usepackage{cleveref}
\usepackage{lipsum}
\usepackage{listings}
\ifnum\pdfshellescape>0
  \usepackage{minted}
\else
  \usepackage[verbatim]{minted}
\fi

\usepackage{threeparttable}

\usepackage{pgfplots}
\pgfplotsset{compat=1.18}
\usepgfplotslibrary{statistics}

\usetikzlibrary{positioning, fit, arrows.meta, shadows.blur}

\newcommand{\step}[2]{%
  \tikz[baseline=(text.base)]{
    \node[inner sep=1pt, font=\small, align=left] (text) {
      \begin{tikzpicture}[baseline=(num.base)]
        \node[draw, circle, fill=white, inner sep=1pt, font=\small] (num) {#1};
      \end{tikzpicture}\hspace{3pt}#2
    };
  }%
}

\newtheorem{definition}{Definition}
\newtheorem{protocol}{Protocol}
\newtheorem{theorem}{Theorem}

\renewcommand{\paragraph}[1]{\smallskip\noindent\textbf{#1}}

\newcommand{\zkzt}{\mathsf{Prezta}}
\newcommand{\name}{\textsc{Prezta}\xspace}

\newcommand{\gen}{\mathsf{Gen}}
\newcommand{\sign}{\mathsf{Sign}}
\newcommand{\idsign}{\mathsf{IDSign}}
\newcommand{\policysign}{\mathsf{PolicySign}}
\newcommand{\prove}{\mathsf{Prove}}

\newcommand{\verify}{\mathsf{Verify}}

\newcommand{\pk}{\mathsf{pk}}
\newcommand{\sk}{\mathsf{sk}}
\newcommand{\vk}{\mathsf{vk}}
\newcommand{\info}{\mathsf{info}}

\newcommand{\G}{\mathcal{G}}
\newcommand{\V}{\mathcal{V}}
\renewcommand{\P}{\mathcal{P}}

\newcommand{\A}{\mathcal{A}}

\newcounter{mytable}
\def\mytable{\begin{centering}\refstepcounter{mytable}}
\def\endmytable{\end{centering}}

\newcounter{myfig}
\def\myfig{\begin{centering}\refstepcounter{myfig}}
\def\endmyfig{\end{centering}}

\newif\ifcameraready
\camerareadyfalse

\newif\ifeprint
\eprinttrue

\newif\ifshowcomment
 \showcommenttrue
\showcommentfalse 

\ifshowcomment

\newcommand{\yupeng}[1]{{\color{red}[Yupeng: #1]}}
\newcommand{\zhongjing}[1]{{\color{blue}[Zhongjing: #1]}}
\newcommand{\osaid}[1]{{\color{purple}[Osaid: #1]}}
\newcommand{\nikita}[1]{{\color{green}[Nikita: #1]}}

\else

\newcommand{\yupeng}[1]{}
\newcommand{\zhongjing}[1]{}
\newcommand{\osaid}[1]{}
\newcommand{\nikita}[1]{}
\fi

\newtheorem*{rep@theorem}{\rep@title}
\newcommand{\newreptheorem}[2]{%
\newenvironment{rep#1}[1]{%
 \def\rep@title{#2 \ref{##1}}%
 \begin{rep@theorem}}%
 {\end{rep@theorem}}}
\newreptheorem{theorem}{Theorem}

\usepackage[available]{usenixbadges}

\begin{document}

\pagestyle{empty}

\date{}

\title{Prezta: Provable Remote Execution of Zero-Trust Authorization using SNARKs}

\author{
{\rm Zhongjing Wei*}\\
University of Illinois Urbana-Champaign\\
zwei26@illinois.edu
\and
{\rm Osaid Muhammad Ameer*}\\
University of Illinois Urbana-Champaign\\
oameer2@illinois.edu
\and
{\rm Nikita Borisov}\\
University of Illinois Urbana-Champaign\\
nikita@illinois.edu
\and
{\rm Yupeng Zhang}\\
University of Illinois Urbana-Champaign\\
zhangyp@illinois.edu
}

\maketitle

\renewcommand{\thefootnote}{\fnsymbol{footnote}}
\footnotetext[1]{* The authors contribute equally to this paper.}
\renewcommand{\thefootnote}{\arabic{footnote}}

\begin{abstract}
Modernizing the security of \emph{operational technology} systems that control critical infrastructure has become a pressing challenge. Because edge devices have limited capabilities,  modernization has relied on \emph{application gateways} that interface with identity management systems and enforce access policies. These gateways are powerful enough to perform complex authorization decisions and support zero-trust architectures but create major deployment and management burdens: they must be collocated with remote, distributed edge devices, kept up to date with security patches, and managed with minimal downtime.

We propose  Provable Remote Execution of Zero-Trust Authorization (\name), an architecture that eliminates these gateways by evaluating policies within a zero-knowledge virtual machine (zkVM) running on the client. The zkVM  produces a succinct proof of authorization that edge devices can verify efficiently,  extending the zero-trust security envelope to the edge. Policies and identity management schemes can evolve without updating edge devices.

To demonstrate the feasibility of \name, we implement a prototype, built using the RISC Zero zkVM, that supports XACML 3.0 policies and JWT identity claims. While zkVMs introduce substantial proof overhead, we mitigate this by compiling policies to Rust code and pre-compiling regular expressions. Combined with optimized signature verification and JWT parsing, these measures reduce prover time by more than an order of magnitude. Our compiler correctly implements 83\% of the XACML 3.0 conformance suite, with proof generation completing in tens of seconds on a desktop. Verification, by contrast, takes only tens of milliseconds—fast enough for even resource-constrained edge devices.


\end{abstract}

\section{Introduction}\label{sec:intro}

Operational technology (OT) networks~\cite{nist:sp800-82} provide a networked interface to physical devices, such as those used in manufacturing, energy systems, and various forms of infrastructure, from water treatment to traffic control. Devices in these networks typically have rudimentary built-in security protections, owing to their limited computational capabilities and upgrade cycles that can range to several decades~\cite{it-ot-convergence}. Yet the security of many of these systems is paramount due to the critical importance of the systems and infrastructure that they monitor and control~\cite{stuxnet,Lee2016UkrainianPowerGrid}. This has led to efforts to extend modern security practices and, in particular, zero-trust architectures~\cite{Rose2020_NIST_ZTA}, to OT networks~\cite{ztot1,ztot2,ztot3}.

A key tool in this effort has been the use of \emph{application gateways}. These gateways implement modern security practices, such as the use of sophisticated role- and attribute-based policies~\cite{rbac,abac}, and interfacing with identity management systems that can implement two-factor authentication. These gateways, however, become a critical component, and their failures or compromise can result in a loss of availability or a security breach, respectively. The network linking the gateways to OT devices is implicitly trusted, so in distributed infrastructure, many gateways are needed to collocate them with the OT devices. This creates a significant management problem, where gateways must be kept secured, patched against vulnerabilities, and available.  

\paragraph{Our Contributions.}
To mitigate this management problem, we propose an alternate architecture, \emph{Provable Remote Execution of Zero-Trust Authorization}, or \name\footnote{Our code is available at \url{https://github.com/walotta/ZK\_Zero\_Trust} and \url{https://github.com/osaidameer/xacml-to-rust}}. Rather than having gateways evaluate complex authorization policies, they are instead executed remotely (typically at the client) inside a zero-knowledge virtual machine (zkVM)~\cite{ben2013snarks,ben2014scalable,risc0}. The zkVM produces a succinct proof, also known as a SNARK, that the policy was correctly evaluated and that an access request is correctly authorized, which can be verified at a very low computational cost. This puts policy enforcement within reach of even low-capability OT devices. Moreover, the flexibility of zkVM means that new policies, policy types, and identity management systems, can all be introduced without any upgrades to the policy enforcement system.
We formalize the algorithms and the security definitions of \name, develop our construction, and prove its security based on the knowledge soundness of the SNARK and the unforgeability of the digital signature. See Section~\ref{sec:arch} for more details. Note that our architecture only utilizes the soundness and succinctness of SNARK, but not zero-knowledge. In fact, the term ``zkVM'' refers to a SNARK for a virtual machine not necessarily with zero-knowledge as a convention in the community.  


To demonstrate the feasibility of \name, we have developed a prototype implementation based on the Risc Zero zkVM~\cite{risc0}. Our prototype supports policies written in XACML~\cite{xacml3} and identity claims supported by a JSON Web Token (JWT)~\cite{jwt}. We compile the XACML policies directly to Rust programs that make policy decisions. The programs take as input an access request, a JWT, and context attributes, and produce a decision. The Risc Zero tools then compile and execute the programs inside the zkVM, producing a proof of authorization. The proof can then be used to verify that a request is authorized. The verification confirms that the correct policy program was executed on the correct inputs and produced the correct output. Our compiler supports most of the XACML 3.0 conformance suite~\cite{authzforce2025}. 

The use of zkVM to execute the policies does introduce significant overhead, and we therefore explored several optimizations in our implementation. We optimize regular expression evaluation by precompiling and serializing the corresponding DFA; we also simplify JWT parsing and use precompiled bigint operations for RSA verifications. Together, these reduce the zkVM execution time, needed to produce a SNARK, from about 20 minutes to under 30 seconds. The cost remains dominated by the RSA verification, so using state-of-the-art RSA verification circuits~\cite{woo2025efficient} the policy execution can be lowered to a median of 7 seconds on our test set. 

Our contributions are summarized below:
\begin{itemize}[leftmargin=*]
    \item We propose a new architecture, \name, for authorizations in OT systems utilizing SNARKs and digital signatures. Our architecture allows OT devices to support complex policies with limited computational resources, and to update policy and identification frameworks without software upgrades on the devices. 

    \item We formalize our architecture with security definitions, and prove that our construction is secure based on the security of SNARKs and digital signatures. 
    
    \item 
    We instantiate the SNARK with a zkVM~\cite{risc0}, and we build a compiler that automatically compiles policies in XACML to Rust, which is supported by the zkVM. Our compiler is able to support 323 out of 389 test cases in the XACML 3.0 conformance suite. 
    
    \item We propose/integrate several optimizations to improve the prover time of the zkVM for the authorization policies. We propose a new method to efficiently check regular expressions in zkVM. We also utilize an existing library for JWT parsing, and the precompiled instructions of \texttt{bigint} for RSA signature verification in zkVM. These optimizations improve the prover speed by nearly 50$\times$. 
    
    \item Finally, we fully implement \name and the end-to-end prover time of most policies in the dataset is around 28 seconds. The prover time can be potentially reduced to around 7 seconds with a special proof for RSA signature verification outside zkVM. With ongoing efforts to improve zkVM performance~\cite{ethproofs}, we believe the prover time can become very practical in the near future. 
\end{itemize}

\subsection{Related Work}

Zero-knowledge proofs have been long used for authentication in the context of anonymous credentials~\cite{brands2000rethinking,camenisch2001efficient}. These systems would allow a client to prove a predicate over attributes in their credential while remaining anonymous. Anonymous credentials have traditionally been implemented using custom zero-knowledge proofs which are relatively expensive to verify and lack the flexibility of our architecture. Several recent papers have used zkSNARKs to bootstrap anonymous credentials from traditional, identity-revealing ones~\cite{rosenberg2023zkcreds}, frequently targeting blockchain applications~\cite{baldimtsi2024zklogin,park2025beyond,aptos_keyless}. The focus is largely on identification, and supporting unlinkability, whereas our focus is on supporting flexible and evolvable authorization policies,  and not privacy, and our target is OT networks and devices. Another novel aspect of our work is the use of zkVMs for greater policy flexibility and ease of implementation, as compared  with circuit-based SNARKs in the previous work.

There is some early work on using zkVMs for authorization and authentication; Moser et al.~\cite{moser2025privacy} proposed a scheme using zkVM to verify digital signatures to authorize smart contract operations on blockchain, while Bonsai Pay~\cite{bonsai} is an application developed by Risc Zero to integrate OpenID with blockchain applications. These systems take a minute or longer to generate their proofs, highlighting the importance of our optimizations. A number of previous systems have proposed using a virtual machine to implement access restrictions~\cite{bakir2021caplets,borisov2002active,cao2024stateful}; their focus was largely on providing least-privilege in delegation contexts, rather than flexible policy implementation, and they used classic (non-verifiable) VMs, which would need to be run either at the end point or a trusted server, both of which are impractical in OT networks. 



\section{Preliminaries}\label{sec:prelim}

\subsection{Application Gateways}

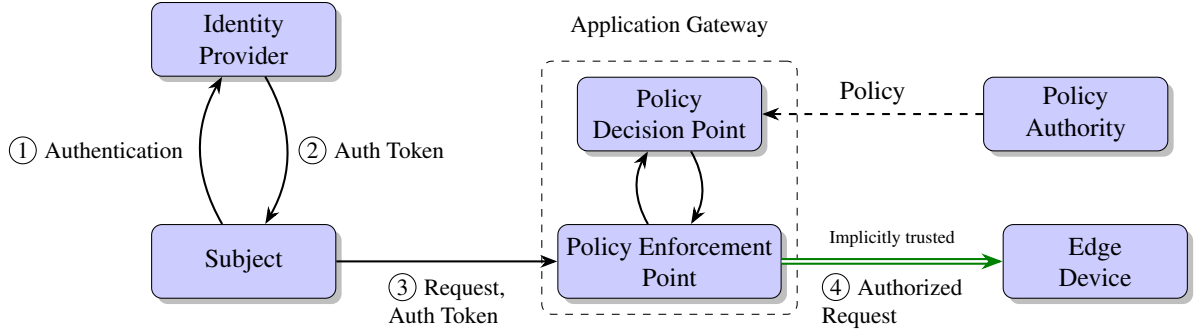
\begin{figure*}[t!]
    \centering
    \resizebox{0.9\textwidth}{!}{
    \begin{tikzpicture}[
        node distance=2cm,
        block/.style={
            draw, rectangle, rounded corners,
            align=center, minimum width=2.5cm, minimum height=1cm,
            fill=blue!20, drop shadow,
        },
        arrow/.style={-Stealth, thick},
        label/.style={align=left},
        ]
        \node[block] (user) {Subject};
        \node[block] (idp) [above=2cm of user] {Identity\\Provider};
        \node[block] (pep) [right=3cm of user] {Policy Enforcement\\Point};
        \node[block] (engine) [above of=pep] {Policy\\Decision Point};
        \node[block] (device) [right=3cm of pep] {Edge\\ Device};
        \node[block] (admin) [right=3cm of engine] {Policy\\Authority};
        
        \node[draw, dashed, rounded corners, fit=(engine) (pep), inner sep=6pt] (appgw) {};
        \node[above=0.2cm of appgw, font=\small] {Application Gateway};
        
        \draw[arrow] (user) to[bend left] node[left, label] {\step{1}{Authentication}} (idp);
        \draw[arrow] (idp) to[bend left] node[right, label] {\step{2}{Auth Token}} (user);
        \draw[arrow] (user) -- node[below, label] {\step{3}{Request,\\Auth Token}} (pep);
        \draw[arrow] (pep)to[bend left] (engine); 
        \draw[arrow] (engine)to[bend left] (pep);
        \draw[arrow, double, double distance=1pt, draw=green!50!black] (pep) --
            node[below, label] {\step{4}{Authorized\\Request}}
            node[above, font=\scriptsize, yshift=2pt] {Implicitly trusted}
            (device);
        \draw[arrow, dashed] (admin) -- node[above, label] {Policy} (engine);
    \end{tikzpicture}
    }
    \caption{Zero Trust architecture}
    \label{fig:zero-trust}
\end{figure*}

Classically, operational technology (OT) environments relied on network isolation security as the primary security mechanism, with edge devices implementing rudimentary, if any, security protections. Recent efforts to modernize security practices have included the deployment of zero-trust architectures, with sophisticated identity management and policies, to OT environments, with the help of \emph{application gateways}. Such gateways  bridge the isolated OT network and enterprise networks and/or the Internet. Requests sent from outside the OT network are intercepted by the gateway and checked for correspondence with a policy before being forwarded to the device. The policy can be sophisticated, such as RBAC, ABAC, and NGAC, and requests can be authenticated using sophisticated and modern identity management (IdM) systems, including reliance on an identity provider, the use of two-factor authentication, etc. 


The gateway thus enables modern zero-trust security practices, but it becomes a new critical component in the OT environment, as all interaction with OT devices must happen via the gateway. As a result, gateways must typically be deployed at every site in an OT environment, both to ensure availability during periods of remote disconnection, and to avoid extending insecure OT networks beyond the premises. At the same time, the gateways are security-critical, since a vulnerability in a gateway allows unrestricted access to the edge devices it protects; e.g., Cesarano and Natella catalog many previously discovered vulnerabilities in application gateways~\cite{Cesarano-Natella}. Thus, application gateways must be carefully managed and kept up-to-date on security patches, while at the same time minimizing downtime due to potential upgrades~\cite{ICSPatch}. 



\subsection{Zero-Trust Authorization Flow}\label{subsec:ZT-auth-flow}

To discuss the authorization flow, we will be using the following terms. Most of these are adapted from the Zero-Trust Architecture specification~\cite{Rose2020_NIST_ZTA}, but slightly adapted for presentation in our system. 
\begin{itemize}[leftmargin=*]
  \item An \emph{Edge Device} is an OT component that performs or monitors a physical action. This could be a relay in an electric substation, a PLC controller for a manufacturing device, or a traffic sensor. 
  \item A \textbf{Subject} is an entity that needs to interact with the edge device by sending it \emph{requests}, such as a maintenance technician, or a monitoring service. We will typically assume that the subject is using 
  a COTS computer system to initiate its requests.
  \item An \textbf{Identity Provider (IdP)} authenticates subjects and assigns them identities, roles, or group membership, that can be used for authorization. The IdP issues the subject with an \textbf{authentication token} that can be used to prove their identity.
  \item A \textbf{Policy Authority (PA)} sets an authorization policy for the system and makes updates to it as needed.
  \item A \textbf{Policy Decision Point (PDP)} decides whether a request from a subject is authorized to perform a request.
  \item A \textbf{Policy Enforcement Point (PEP)} interacts with the PDP to determine whether a request from a subject is properly authorized and takes action to either permit the request to reach the edge device, or block it.
\end{itemize}


We describe the authorization flow in an application gateway, as shown in Figure~\ref{fig:zero-trust}. The PA specifies a policy and deploys it on the PDP, running inside the application gateway. The subject authenticates with the IdP and obtains a token, which is forwarded along with the request to the gateway. The PEP, also inside the gateway, communicates with the PDP and, based on the decision, either forwards the request through to the edge device or denies it. Note that beyond the PEP, no further authentication or authorizations are made, and the communication from the gateway is implicitly trusted by the edge device. Denied requests (or all requests) can be sent to an audit log (not shown); we discuss audit in \Cref{sec:audit}.   

The PDP makes use of the following information when making a decision:
\begin{itemize}[leftmargin=*]
    \item Subject attributes, e.g., usernames, email addresses, roles, etc., as authenticated by the identity provider
    \item Request attributes, e.g., requested operation, and parameters
    \item Edge device attributes, e.g., device type, location, or label
    \item Context, e.g., the current date and time, or the IP address of the requester
\end{itemize}

In some enterprise deployments of ZTA, the IdP and PDP are combined into a single \emph{identity and access manager} (IAM) which issues an \emph{authorization token} for a particular request. Involving an IAM at every request, however, is not appropriate for OT environments where delays and availability issues associated with IAM may prevent critical operation. Even the IdP need not be contacted with every request, as authentication tokens can be reused until their expiry. The expiration times are typically minutes to hours, but can be tuned to trade off the performance, availability, and security risks. 

\paragraph{JWT.} In this paper, we use the JSON Web Tokens (JWT) as the authentication token, but our scheme can be generalized to other formats. A JWT consists of three fields: \texttt{\textless Header\textgreater, \textless Payload\textgreater, \textless Signature\textgreater}. The header includes the algorithm of the digital signature (e.g.\ RSA-SHA256) and other meta-data. The payload includes the subject’s role, user-group, issued timestamp and expiration timestamp. The signature is a digital signature signing both the header and the payload in base64. We use RSA-2048 and SHA-256 in this paper. An example of JWT is provided in~\Cref{subsubsec:jwt}.

    
    

\subsection{Cryptographic Primitives}

\paragraph{Digital signature.} A digital signature scheme consists of the following algorithms: $\gen(1^\lambda)\to (\sk,\pk)$, $\sign(m,\sk)\to\sigma$, $\verify(m,\sigma,\pk)\to\{0,1\}$.
A signature scheme is correct if for all $(\sk,\pk)\gets\gen(1^\lambda)$, all $\sigma\gets\sign(m,\sk)$, $\verify(m,\sigma,\pk)=1$. 
It is unforgeable if for all PPT adversary $\A$, $(\sk,\pk)\gets\gen(1^\lambda)$, $(m^*,\sigma^*)\gets\A^{\sign(\cdot)}(\pk)$, $\Pr[\verify(m^*,\sigma^*,\pk)=1 \land m^*\notin Q]\le\mathsf{negl}(\lambda)$, where $Q$ is the set of messages queried to the signing oracle $\sign(\cdot)$. 

\paragraph{SNARK.}  A Succinct non-interactive argument of knowledge (SNARK) allows a prover to convince a verifier that a statement is true. It consists of three algorithms ($\mathcal{G}, \mathcal{P}, \mathcal{V}$), and satisfies completeness and soundness. See Appendix~\ref{app:prelim} for the formal definitions. In our construction, we instantiate the SNARK with a zkVM, which represents the relation by a program $P$. The statement is $y=P(x,w)$, denoting the output of running the program on the input and the witness. zkVM also supports the generation of $\vk_P$ for a program, and we abuse the notation to denote it as $\vk_P\gets\G(1^\lambda,P)$. In our construction, we use the zkVM of RISC Zero~\cite{risc0}, the backend of which is STARK~\cite{stark} with a transparent setup. Therefore, $\pk$ simply consists of a hash function and some public parameters that do not depend on $P$.

\section{Architecture Design}\label{sec:arch}
In this section, we present the security definitions and the architecture of our new design, \name. We also provide a generic construction using SNARKs and digital signatures. 

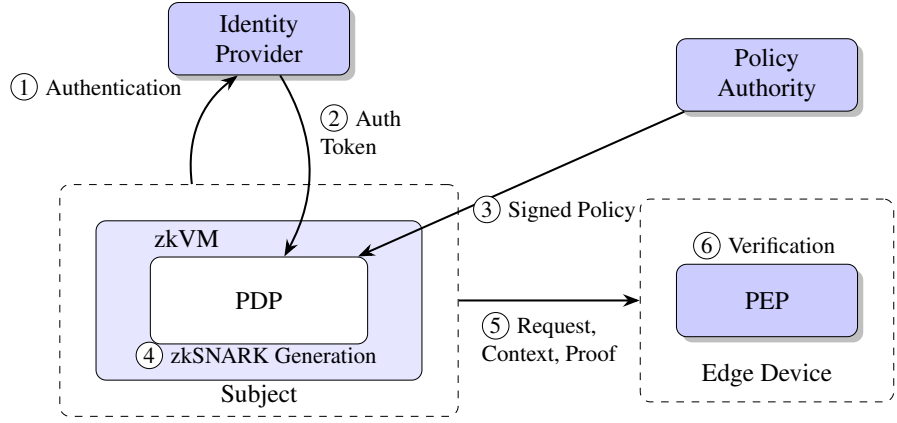
\begin{figure*}[t!]
    \centering
    \resizebox{0.68\textwidth}{!}{
    \begin{tikzpicture}[
        node distance=2cm,
        block/.style={
            draw, rectangle, rounded corners,
            align=center, minimum width=2.5cm, minimum height=1cm,
            fill=blue!20, drop shadow,
        },
        arrow/.style={-Stealth, thick},
        label/.style={align=left},
        ]
        \node[draw, rectangle, dashed, rounded corners, minimum width=5.5cm, minimum height=3.2cm] (user) {};
        \node[anchor=south, yshift=0cm] at (user.south) {Subject};
        \node[
            draw, rectangle, rounded corners, align=center,
            minimum width=4.5cm, minimum height=2.2cm, fill=blue!10
        ] (zkvm) at (user) {};
        \node[anchor=north, xshift=-1cm] at (zkvm.north) {zkVM};
        \node[
            draw, rectangle, rounded corners, align=center,
            minimum width=3cm, minimum height=1.2cm, fill=white
        ] (pdp) at (zkvm) {PDP};
        \node[block] (idp) [above=1.5cm of user] {Identity\\Provider};
        
        \node[block] (pep) [right=3cm of user] {PEP};
        \node[draw, rectangle, dashed, rounded corners, align=center, minimum width=3.5cm, minimum height=2.8cm] (device) at (pep) {};
        \node[anchor=south, yshift=0.1cm] at (device.south) {Edge Device};
        \node[block] (admin) [above=1.2cm of device] {Policy\\Authority};

        \draw[arrow] (user) to[bend left] node[left, label, yshift=0.5cm] {\step{1}{Authentication}} (idp);
        \draw[arrow] (idp) to[bend left] node[right, label, yshift=0.5cm] {\step{2}{Auth\\Token}} (pdp);
        \draw[arrow] (admin) -- node[below,label, xshift=0.5cm] {\step{3}{Signed Policy}} (pdp);
        \node[anchor=south] at (zkvm.south) {\step{4}{zkSNARK Generation}};
        
        \draw[arrow] (user) -- node[below, label] {\step{5}{Request,\\Context, Proof}} (device);

        \node[anchor=north, yshift=0.6cm] at (pep.north) {\step{6}{Verification}};
    \end{tikzpicture}
    }
    \caption{Our new \name architecture}
    \label{fig:zkzt}
\end{figure*}

\subsection{Security Model} 








\paragraph{Security Model.} We assume that the subject is untrusted and may be compromised by an adversary to make unauthorized access to the device. Both the IdP and the PA are assumed to be trusted and their public keys are publicly known. We do not consider adversaries that compromise the network communications. Network-level threats (e.g., man-in-the-middle, replay attacks) are defended against with orthogonal mechanisms and are out of the scope of this paper. 

Informally speaking, our \name architecture guarantees that the subject is granted the access of the device if and only if the subject possesses a valid token issued by the IdP, and the subject's attributes satisfy the policy authorized by the PA. 
We formally define \name below:



\begin{definition}\label{def:zkzt}
    A \name scheme is a tuple of algorithms:
    \begin{itemize}[leftmargin=*]
        \item $\zkzt.\gen(1^\lambda)\to (\sk_\mathsf{IdP}, \pk_\mathsf{IdP},\sk_\mathsf{PA}, \pk_\mathsf{PA}, \pk_\mathsf{SNARK})$: the algorithm takes the security parameter as input, and outputs the pairs of secret key and public key of IdP and PA, as well as the public parameters of the SNARK. They can be generated separately, but we combine their generations into one algorithm for simplicity. 
        
        \item $\zkzt.\idsign(\sk_\mathsf{IdP},\mathsf{att})\to \sigma_\mathsf{att}$: the algorithm is executed by the IdP. It takes the secret key of IdP and the attributes of the subject and outputs an authorization token.
        
        \item $\zkzt.\policysign(\sk_\mathsf{PA}, P)\to (\vk_P, \sigma_P)$: the algorithm is executed by the PA. It signs the access control policy of the edge device and outputs the verification key and the signature of the policy $P$. 
        
        \item $\zkzt.\prove(P, \mathsf{att}, Req, \info, \sigma_{P}, \sigma_\mathsf{att}, \pk_\mathsf{SNARK})\to \pi$: the algorithm is executed by the subject. For a request $Req$ to access the edge device, the algorithm takes as input the policy, the subject's attributes, the subject's token, the signature of the policy and other public information, and computes a proof. Here $\info$ denotes device attributes (e.g., type, location) and context (e.g., date, IP address). 
        
        \item $\zkzt.\verify(\pi, Req, \info, \vk_P, \sigma_P, \pk_\mathsf{PA})\to \{0,1\}$: upon receiving a request together with a proof, the system runs this verification algorithm and outputs 0 or 1. 
    \end{itemize} 

    \paragraph{Completeness.} A \name scheme is complete if for all $(\sk_\mathsf{IdP}, \pk_\mathsf{IdP},\sk_\mathsf{PA}, \pk_\mathsf{PA},\pk_\mathsf{SNARK})\gets\zkzt.\gen(1^\lambda)$,       $\sigma_\mathsf{att}\gets\zkzt.\idsign(\sk_\mathsf{IdP},\mathsf{att})$,   $(\vk_P, \sigma_P)\gets \zkzt.\policysign(\sk_\mathsf{PA}, P)$:
    \begin{equation*}    
        \Pr \left[ 
        \begin{tabular}{l}
        $P(\mathsf{att}, Req, \info) = 1 \land $\\
         $\pi\gets\zkzt.\prove(P, \mathsf{att}, Req, \info, $\\
         \hspace{7em}$\sigma_{P}, \sigma_\mathsf{att}, \pk_\mathsf{SNARK})$\\
         : $\zkzt.\verify(\pi, Req, \info, \vk_P, \sigma_P, \pk_\mathsf{PA})=1$
        \end{tabular}
        \right] =1, 
    \end{equation*}
    where $P(\mathsf{att}, Req, \info)$ denotes the output of the policy on the request, the subject's attributes and the public information. 

    \paragraph{Security.} A \name scheme is secure if for any PPT adversary $\A$, for any $(\sk_\mathsf{IdP}, \pk_\mathsf{IdP},\sk_\mathsf{PA}, \pk_\mathsf{PA},\pk_\mathsf{SNARK})\gets\zkzt.\gen(1^\lambda)$. the following probability is $\le \mathsf{negl}(\lambda)$:
    \begin{equation*}    
        \Pr \left[ 
        \begin{tabular}{l}
        $(Req,\mathsf{att}, \sigma_\mathsf{att},\pi, P,\vk_P,\sigma_P)\gets\A^{\mathsf{IDSign}(\cdot),\mathsf{PolicySign}(\cdot)}($\\ \hspace{11em}$\pk_\mathsf{IdP},\pk_\mathsf{PA},\pk_\mathsf{SNARK})$,\\
        
         : $(\zkzt.\verify(\pi, Req, \info, \vk_P, \sigma_P, \pk_\mathsf{PA})=1$\\
         and $P(\mathsf{att},Req,\info)\neq 1)$\\
         $\lor (\verify(\mathsf{att}, \sigma_\mathsf{att},\pk_\mathsf{IdP})= 1$ and $\mathsf{att}\notin Q_\mathsf{att})$\\
         $\lor (\verify(\vk_P, \sigma_P,\pk_\mathsf{PA})= 1$ and $\vk_P\notin Q_P)$.
        \end{tabular}
        \right]
    \end{equation*}
where $Q_\mathsf{att}$ is the set of messages queried to the signing oracle $\mathsf{IDSign}(\cdot)$, and $Q_P$ is the set of messages queried to $\mathsf{PolicySign}(\cdot)$.

\end{definition}

\paragraph{Unlinkability.} An additional property that can be achieved by our design is the unlinkability of the access requests. This can be enabled by the zero-knowledge property of the zkSNARK, and by changing the algorithms to hide some information in $Req$ and $\info$ such as the IP address of the subject. We primarily focus on the security of the \name scheme and leave the formal definition and construction of unlinkability as a future work.

\subsection{Our \name System}

\begin{figure}[t!]
\centering{\centering
\framebox{\parbox{.99\linewidth}{
\begin{protocol}[\name scheme]\label{prot:zkzt}
\text{   }
\begin{itemize}[leftmargin=*]
        \item $\zkzt.\gen(1^\lambda)\to (\sk_\mathsf{IdP}, \pk_\mathsf{IdP},\sk_\mathsf{PA}, \pk_\mathsf{PA}, \pk_\mathsf{SNARK})$
        \begin{enumerate}
            \item $(\sk_\mathsf{IdP}, \pk_\mathsf{IdP})\gets\mathsf{signature}.\gen(1^\lambda)$
            \item $(\sk_\mathsf{PA}, \pk_\mathsf{PA})\gets\mathsf{signature}.\gen(1^\lambda)$
            \item $\pk_\mathsf{SNARK}\gets\mathsf{SNARK}.\G(1^\lambda)$
        \end{enumerate}
        
        \item $\zkzt.\idsign(\sk_\mathsf{IdP},\mathsf{att})\to \sigma_\mathsf{att}$
        \begin{enumerate}
            \item $\sigma_\mathsf{att}\gets\sign(\mathsf{att},\sk_\mathsf{IdP})$
        \end{enumerate}
        
        \item $\zkzt.\policysign(\sk_\mathsf{PA}, P)\to (\vk_P, \sigma_P)$
        \begin{enumerate}
            \item $P$ hardcodes $\pk_\mathsf{IdP}$, or a list of acceptable public keys of IdPs. 
            \item $\vk_P\gets \mathsf{SNARK}.\G(1^\lambda, P)$
            \item $ \sigma_P\gets\sign(\vk_P,\sk_\mathsf{PA})$
        \end{enumerate}
        
        \item $\zkzt.\prove(P, \mathsf{att}, Req, \info, \sigma_{P}, \sigma_\mathsf{att}, \pk_\mathsf{SNARK})$ $\to \pi$
        \begin{enumerate}
            \item $\pi\gets\P(P, \mathsf{att},Req,\info,\sigma_\mathsf{att}, \pk_\mathsf{SNARK})$ where the statement of the SNARK is

            $\mathcal{R} = \{(P,Req,\info; \mathsf{att},\sigma_\mathsf{att}) | P(\mathsf{att},Req,\info)$ $ = 1 \land \mathsf{signature}.\verify(\mathsf{att},\sigma_\mathsf{att}, \pk_\mathsf{IdP})=1\}$.

        \end{enumerate}
        
        \item $\zkzt.\verify(\pi, Req, \info, \vk_P, \sigma_P, \pk_\mathsf{PA})\to \{0,1\}$

        \begin{enumerate}
            \item Obtain $\info$
            \item Check if $\mathsf{signature}.\verify(\vk_P,\sigma_{P},\pk_\mathsf{PA})=1$. This only needs to be checked once for a policy.
            \item Check if $\mathsf{SNARK}.\V(Req,\info, \vk_P, \pi)=1$
            \item Output 1 if all checks pass; otherwise, output 0.

        \end{enumerate}
\end{itemize}
\end{protocol}}}}

\end{figure}
Our new \name architecture is presented in~\Cref{fig:zkzt}, and the generic protocol of our scheme is presented in~\Cref{prot:zkzt}. It utilizes the SNARK scheme and the digital signature scheme as defined in~\Cref{sec:prelim}. As shown in the figure, we completely remove the gateway from the architecture, and utilize SNARK to enforce the policy. The overhead of policy evaluation and proof generation is shifted to the subject, thus addressing the update and maintenance issues of gateways and edge devices in the existing zero trust architecture (\Cref{fig:zero-trust}). In the protocol, we hardcode $\pk_\mathsf{IdP}$, or a list of acceptable public keys of IdPs, in the policy, which is signed by the PA. We can also make $\pk_\mathsf{IdP}$ as a public input known by the edge device, and there is no essential difference between the two choices. 

In Appendix~\ref{app:proof}, we provide a proof sketch for the following theorem:
\begin{theorem}\label{thm:main}
    \Cref{prot:zkzt} is a complete and secure \name scheme by~\Cref{def:zkzt}. 
\end{theorem}


\section{System Design}\label{sec:system}

Following the generic construction in the previous section, we further present the detailed design of \name to make it concretely efficient in practice. 

We instantiate the zkSNARK with a zkVM. The primary motivation for adopting a zkVM rather than designing custom circuits is that the authorization logic involves several components with high implementation complexity, including regular-expression based string matching, JSON Web Token (JWT) parsing, and RSA signature verification. Constructing hand-optimized circuits for such components is cumbersome, error-prone, and expensive to maintain, especially when policy rules evolve frequently.

A zkVM allows us to directly reuse well-engineered library implementations of these components in Rust, maintaining high code readability, auditability, and testability. At the same time, with targeted optimizations on these components (detailed in Section~\ref{subsec:zkvm_optimizations}), the zkVM proving overhead is significantly reduced, making its proving cost comparable to that of tailored circuits for our authorization workloads. This design enables both practical deployment and security soundness without sacrificing performance.

Due to the use of zkVM, in particular RISC Zero~\cite{risc0}, the policy will be compiled to a Rust code at the first step; then the Rust code will be converted to machine code (RISC-V in our case. This machine code is called the IMAGE) and the zkVM module will also output a commitment of machine code (this commitment is called the IMAGE\_ID). The commitment serves as $\vk_P$ defined previously, and when performing the verification, the verifier will check if the proof is aligned with the commitment of machine code. This technology can ensure the code in zkVM is identical to the trusted policy logic approved by the policy authority, thus the subject cannot change the policy.

The \name system works in three main stages:
\begin{enumerate}[leftmargin=*]
    \item \textbf{Policy Deployment:} The policy admin compiles the authorization policy in XACML format into Rust code, which is further compiled into RISC-V machine code (IMAGE). A commitment to this machine code (IMAGE\_ID) is generated for later verification. Then the policy admin deploys the IMAGE\_ID to the edge device, which stores it for future access requests. Also the policy admin will share the policy code (IMAGE) with authorized subjects who may request access to resources protected by this policy.

    \item \textbf{Proof Generation:} When a subject requests access to a resource, they provide their attributes (e.g., JWT token, source IP address), environment attributes (e.g., current Unix timestamp) and request context (e.g., requested resource ID and request operation) to the zkVM along with the compiled policy IMAGE. The zkVM executes the policy logic using these attributes and generates a zero-knowledge proof attesting that the access decision is computed correctly according to the policy. Then the subject sends the access request and the generated proof to the edge device for verification.

    \item \textbf{Proof Verification:} The edge device receives the subject's access request and the generated proof. It first verifies that the proof corresponds to the committed machine code (IMAGE\_ID) ) stored locally during policy deployment. Then, it checks the authenticity of all public input attributes against trusted sources before accepting the proof and making the final access decision. These checks include consistency tests between subject-provided values and verifier-observable data when applicable.

\end{enumerate}

\subsection{Policy Compiler}\label{subsec:compiler}

To enable verifiable evaluation of XACML policies in constrained OT environments, we developed a policy compiler that translates XACML policies directly into Rust. The resulting programs take the request and context as input and output whether the source XACML policy would have permitted or denied access. The programs can then be compiled into RISC-V code that can produce a verifiable execution in the Risc Zero virtual machine. The entire flow is illustrated in Figure \ref{fig:pipeline}.

\subsubsection{XACML Background} Before diving into the compiler and its design, we first provide some background on XACML itself. The eXtensible Access Control Markup Language (XACML), is an OASIS standard which is designed to express fine-grained, attribute-based access control policies. XACML enables the combination of role-based permissions with attribute-based evaluation of requests, for example, characteristics of the subject (e.g. role, department), the resource (e.g file type, sensitivity level), the action (e.g. read, write) and the environment (e.g. time of day, network location). 

 At a high level, XACML policies consist of the following elements:
\begin{itemize}[leftmargin=*]
    \item \textbf{PolicySet}: A container that groups policies or other policy sets, defining how their decisions are combined.
    \item \textbf{Policy}: A container that groups rules and defines the overall structure of the authorization logic.
    \item \textbf{Rule}: This is the smallest unit of decision-making in a policy. It typically consists of an \emph{effect} (e.g. Permit or Deny) and logic determining when that effect applies.
    \item \textbf{Target}: This is a filtering mechanism inside a rule or policy that defines to which request a policy or rule applies. This enables quick filtering of applicable policies/rules.
    \item \textbf{Condition}: A boolean expression that evaluates to true for a rule to apply.
    \item \textbf{Combining Algorithm}: These are strategies for resolving conflicts when multiple rules, policies, or policy sets apply;  e.g., deny-overrides or permit-overrides.
\end{itemize}

Together, these components enable XACML to express complex authorization logic while remaining structured and declarative. Listing~\ref{lst:xacml-sample} shows a simplified XACML policy that permits physicians to read a specific patient record, demonstrating how targets and rules are structured in practice.\footnote{%
Our compiler works directly with XACML policies, but for ease of readability, we show policies using the ALFA syntax~\cite{ALFA}, which is much more compact than XACML. ALFA can be compiled directly to XACML using existing tools.}

\begin{listing}[!ht]
\begin{minted}[fontsize=\small, frame=none, breaklines, breakanywhere]{java}
namespace example.policies {
    policy A001_policy {
        apply deny-overrides;
        rule IIA003_rule permit {
            target
                subject.bogus == "Physician"
                and resource.resource_id ==
                    "http://medico.com/record/patient/BartSimpson"
                and (action.action_id == "read" or action.action_id == "write");
        }
    }
    // Attribute definitions
    // ...
}
\end{minted}
\caption{ALFA version of an XACML Policy}
\label{lst:xacml-sample}
\end{listing}




\begin{figure}[htbp]
\centering
\begin{tikzpicture}[node distance=1.2cm, auto]
    \tikzstyle{compile} = [draw, rectangle, align=center, fill=blue!10]
    \tikzstyle{runtime} = [draw, rectangle, align=center, fill=green!10]
    \tikzstyle{input} = [draw, rectangle, align=center, fill=yellow!20]
    \tikzstyle{output} = [draw, rectangle, align=center, fill=orange!20]
    
    \node[compile] (xacml) {XACML};
    \node[compile] (ir) [below=0.5cm of xacml] {Intermediate Representation};
    \node[compile] (code) [below=0.5cm of ir] {Rust Code};
    \node[compile] (compile) [below=0.5cm of code] {RISC-V Binary};
    
    \node[runtime] (execute) [below=0.8cm of compile] {Guest Program \\
    Execution};
    
    \node[input] (jwt) [left=0.65cm of execute, yshift=0.4cm] {JWT};
    \node[input] (request) [left=0.65cm of execute, yshift=-0.4cm] {Req. Attributes};

    \node[output] (decision) [right=0.65cm of execute, yshift=0.4cm] {Decision};
    \node[output] (req) [right=0.65cm of execute, yshift=-0.4cm] {Req. Attributes};
    
    \draw[->, thick] (xacml) -- (ir);
    \draw[->, thick] (ir) -- (code);
    \draw[->, thick] (code) -- (compile);
    \draw[->, thick] (compile) -- (execute);
    
    \draw[->, thick] (request.east) -- (execute.west |- request.east);
    \draw[->, thick] (jwt.east) -- (execute.west |- jwt.east);

    \draw[->, thick] (execute.east |- decision.west) -- (decision.west);
    \draw[->, thick] (execute.east |- req.west) -- (req.west);
    
    \node[draw=green!60, thick, dashed, fit={(execute)}, 
          inner sep=0.5cm, inner xsep=0.45cm, inner ysep=0.5cm,
          label={[anchor=south east]south east:{\small\textit{zkVM Environment}}}] {};
\end{tikzpicture}
\caption{End-to-end pipeline from XACML policy to verifiable execution in zkVM. Compilation stages (blue) transform the policy into executable code, while runtime stages (green) execute within the zkVM environment with request attributes and JWT (yellow) as inputs, with the decision and respective request attributes (orange) as output.}
\label{fig:pipeline}
\end{figure}
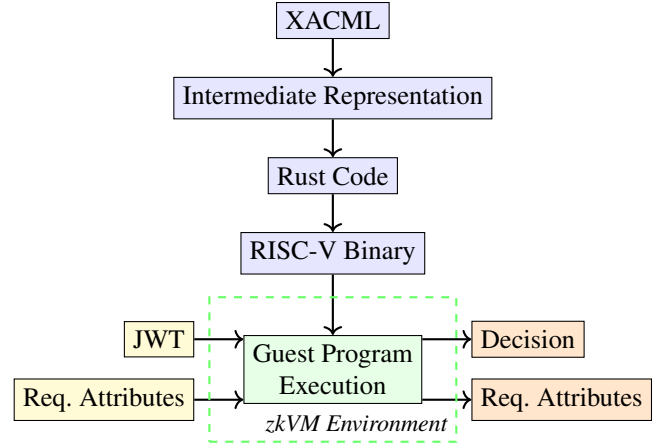


\subsubsection{Type System and Mapping} XACML's type system includes primitives (string, integer, boolean, double), temporal types (date, time, dateTime), durations (dayTimeDuration, yearMonthDuration) and other specialized types. Our compiler maps these to appropriate Rust types: strings map to String, integers to \texttt{i32} or \texttt{i64}, booleans to \texttt{bool}, and temporal types leverage existing crates (chrono for date/time types, iso8601\_duration for duration parsing). Multi-valued attributes (known as bags in XACML terminology) are represented as Rust vectors (e.g. \texttt{Vec<String>}, \texttt{Vec<i32>}).

\par The Input struct is generated dynamically based on all attributes referenced in the policy, with fields strongly typed according to their declared XACML data type. Type safety is currently enforced at compile time, where the Rust compiler validates that all comparisons and operations are type-correct, catching any inconsistencies introduced during code generation. At runtime, the host program is responsible for providing the correctly-typed inputs matching the policy's expectations. Since proof generation requires valid execution, malformed or mistyped inputs will simply fail to produce a valid proof. We do not currently perform type validation at the IR or policy level, assuming input XACML policies are semantically correct. 

\subsubsection{Expression and Code Generation} Logical expressions, including conditions and targets, are generated by recursively traversing the IR. Each operator is mapped to a handler that emits typed Rust code, with handlers recursively processing child nodes in a depth-first manner. This approach naturally mirrors the tree structure of logical expressions in XACML. 

\par Leaf nodes---attributes or constants---are converted into Rust expressions with appropriate type handling. Attribute references are transformed into field accesses on the input structure, with insertion of parsing logic where required for complex data types (e.g. temporal, durations). 

\par Static Jinja templates define the structure of Rule, Policy, and PolicySet functions. Using templates is a design choice that follows the natural structure and hierarchical nature of operations in XACML, resulting in verifiable, deterministic, and human-readable code. These templates function as boilerplate structures of XACML policies, while the expressions inside these structures are dynamically generated, dependent on the logic enclosed in the particular policy. 

\par Combining-algorithm logic is also embedded within these templates through conditional blocks. The Compiler supports all standard XACML 3.0 combining algorithms, such as \texttt{permit-overrides}, in which a single \emph{Permit} decision overrides the results of all sibling policies or rules; its counterpart, \texttt{deny-overrides}, where a single \emph{Deny} decision takes precedence; and \texttt{permit-unless-deny}, which yields an overall \emph{Permit} decision unless any child element evaluates to \emph{Deny}. These implementations follow advice from the XACML 3.0 specification document \cite{xacml3}, and exist as distinct code paths within the template, selected dynamically based on the policy's specified combining algorithm. The same template structure applies to both policy-level rule combination and policyset-level policy combination, with the exception of the \texttt{only-one-applicable} combining algorithm which only applies to \texttt{PolicySets}. Listing~\ref{lst:rust-generated} shows the Rust code generated from the sample policy in Listing~\ref{lst:xacml-sample}, illustrating how the hierarchical policy structure is translated into a series of evaluation functions. The complete code and policy are provided in Appendix \ref{app:codegen}.

\begin{listing}[!ht]
\begin{minted}[fontsize=\small, frame=none, breaklines, breakanywhere]{rust}
// {omitted imports and jwt functions}
fn evaluate_rule_target(inp: &Inputs) -> bool {
    (("Physician" == inp.access_subject_bogus)
        && ("http://medico.com/record/patient/BartSimpson" == inp.resource_resource_id)
        && (("read" == inp.action_action_id) || ("write" == inp.action_action_id)))
}
fn evaluate_rule(inp: &Inputs) -> Result {
    if !evaluate_rule_target(inp) {
        return Result::NotApplicable;
    }
    return Result::Permit;
}
fn evaluate_policy(inp: &Inputs) -> Result {
    // {omitted target evaluation + combining algorithm logic}
}
fn main() {
    let inp: Inputs = env::read();
    let mut decision = match evaluate_policy(&inp) {
        Result::Permit => true,
        _ => false,
    };
    // {omitted jwt verification}
    env::commit(&decision);
    env::commit(&inp);
}
\end{minted}
\caption{Rust code generated from policy defined in Listing \ref{lst:xacml-sample}, illustrating rule evaluation. Combining logic and auxiliary functions are omitted for brevity, complete code and policy are included in Appendix \ref{app:codegen}}
\label{lst:rust-generated}
\end{listing}

\subsubsection{Input/Output}\label{sec:i/o}
 XACML policies, while defining static values to be used in comparisons, also strictly define what fields and their respective data types are expected in the access request. The Compiler leverages this to dynamically define an input structure for each policy. This is accomplished by parsing the policy file, extracting all the \texttt{AttributeDesignator} elements (which specifies which attributes the policy expects to be present in a request) and their data types. The structure of the policy tree is analyzed to determine whether an attribute could be multi-valued (a bag in XACML terminology), in which case it is defined as a vector in the generated Rust code. This approach ensures that the input structure precisely matches the policy's requirements, allowing a request to be read by a host program and passed directly as input to the guest policy code. 

 \par Alongside establishing the request structure, the policy evaluator must also validate that the request originates from an authenticated source. Since our architecture removes traditional components such as the Policy Enforcement Point or Application Gateways, we must supplement our policy evaluator with request-validation capabilities. Each policy program accepts a JWT (our chosen authentication token), which acts as the source of subject attributes as defined in Section \ref{subsec:ZT-auth-flow}, decodes the contents, verifies its signature, and ensures that the subject from the token matches the subject of the request. In our current prototype, the public key of the Identity Provider, used for JWT verification, is embedded into the policy program at compile time. 
 
 \par The decision output is simplified to a binary outcome, where a regular \texttt{Permit} maps to true, while \texttt{NotApplicable}, \texttt{Deny}, and \texttt{Indeterminate} map to false. This simplification is intentional, as our architecture is primarily concerned with whether requests are correctly authorized or denied according to the policy. While \texttt{NotApplicable} and \texttt{Indeterminate} represent aspects of policy evaluation logic rather than explicit access outcomes, their distinction is less critical to our binary authorization model. Nevertheless, the compiler does retain support for XACML's extended result functionality should more granular decision information be needed in future work, and to support the correct implementation of combining algorithms. 

\subsubsection{Validation}\label{subsubsec:compiler_exp} We validate our compiler using the XACML 3.0 conformance test suite \cite{authzforce2025}, which contains 397 mandatory test cases. Each test case in the suite contains three files: policy.xml, request.xml, and response.xml (with a few exceptions). From these 397, 8 lack associated request files or have multiple policy files (linked through \texttt{PolicyIdReference} or \texttt{PolicySetIdReference}, an unsupported feature), leaving 389 testable cases. After filtering for unsupported features,  we evaluate our Compiler on 323 policies. Since the test suite contains the appropriate responses for the specific request input, we can sufficiently validate our compiler implementation. 

\par Our validation process is fully automated, for each test case, both the request and expected response are parsed into JSON, the policy is compiled to Rust, which is finally executed with the test request as input and the resulting decision is compared against the expected response. We pre-process the the policy results to map to the true and false binary output we expect, as mentioned in Section \ref{sec:i/o}. All 323 evaluated test cases pass successfully, demonstrating that the compiler correctly implements XACML semantics for the supported feature set. 


\subsubsection{Limitations and Scope} Our compiler implements most XACML functions---numeric, logical, string, set operations, and regular expressions. We did not implement certain functionality, such as XPath queries, RFC 822 names, or duration arithmetic, since our goal was a prototype compiler that demonstrates feasibility rather than one that achieves full XACML conformance. Nevertheless, the compiler supports a large majority of the conformance suite. A full list of unsupported XACML functions is provided in Appendix~\ref{app:unsupported} for reference.
We note that our compiler must be re-run every time an XACML policy is updated to generate new Rust policy code; however, this introduces negligible overhead, since its execution is typically well under a second. 

\subsection{zkVM and optimizations}\label{subsec:zkvm_optimizations}



We present the zkVM design and optimizations to make \name efficient in this section. 

In \name, the edge device cannot rely solely on a valid SNARK proof for authorization. The prover (i.e., the subject) may generate a valid proof using forged or manipulated input attributes. Therefore, the edge device must validate the authenticity of all public inputs before accepting the proof. To illustrate how different input attributes are treated in the \name system, we classify them based on their visibility and required real-world binding, as summarized in Table~\ref{tab:input_attributes}. 


\begin{table}[t]
\centering
\caption{Example attribute visibility and trusted binding source}
\label{tab:input_attributes}
\begin{threeparttable}

{\small
\begin{tabular}{lcc}
\toprule
\textbf{Attribute} & \textbf{Visibility} & \textbf{Binding Source} \\
\midrule
JWT & Private & -\tnote{1} \\
JWT public key & Public & Committed with zkVM code \\
subject IP address & Public & Network stack \\
Requested resource ID & Public & Request context \\
System timestamp & Public & System clock \\
\bottomrule
\end{tabular}
}

\begin{tablenotes}
\footnotesize
\item[1] Private attributes are validated implicitly inside the zkVM (e.g., JWT claims verified using signature).
\end{tablenotes}

\end{threeparttable}
\end{table}




As we will demonstrate in the experiments, naively implementing the policies in zkVM would introduce a high overhead on the prover time. We introduce several optimizations to reduce the overhead significantly. 

\subsubsection{Regex} 
The first type of function with a high overhead is regular expressions. It appears in some of the policies to perform string match on the username. Listing~\ref{lst:regex_naive} shows an example from case \textit{IIC057} in the XACML 3.0 conformance test suite to check the regular expression of \texttt{J.* K.* Hibbert}. The code was compiled naively into Rust  using the Regex library. It executes the full regex engine to parse and match this pattern at the runtime.

\begin{listing}[!ht]
\begin{minted}[fontsize=\small, frame=none]{rust}
use regex::Regex;
let is_match = Regex::new(r"J.* K.* Hibbert")
                  .unwrap()
                  .is_match("Julius K. Hibbert");
\end{minted}
\caption{Default regex matching in Rust}
\label{lst:regex_naive}
\end{listing}

Running a generic Rust regex engine inside a zkVM significantly increases the prover time due to the heavy Nondeterministic/Deterministic Finite Automaton (NFA/DFA) construction and state transitions performed at runtime. In particular, the dynamic compilation of arbitrary user-provided patterns results in a substantial growth of the number of CPU cycles because the entire automaton evaluation logic must be represented in the zkVM. 

In particular, when running the example in Listing~\ref{lst:regex_naive}, the naive approach performs the following steps at runtime:
\begin{enumerate}[leftmargin=*]
    \item Parse the user-supplied regular expression into an abstract syntax tree (AST);
    \item Translate the AST into an intermediate automaton representation (typically a Thompson NFA);
    \item Execute the NFA (or construct a DFA) dynamically during pattern matching.
\end{enumerate}
When executed inside a zkVM, this entire workflow is recorded in the execution trace, including parsing, automaton construction, and state transitions. Since every step must be represented by arithmetic constraints in the proof system, dynamic pattern compilation leads to substantial constraint growth and proving overhead. 

\paragraph{Our approach.} In \name, the policy is statically known at the time of deployment. Therefore, regexes can be precompiled into deterministic finite automata (DFA), and the resulting transition tables can be embedded directly into the zkVM program as serialized binary artifacts. During proof generation, the zkVM performs only table-driven state transitions, eliminating both pattern parsing and automaton construction overhead.
In particular, the workflow consists of three stages:
\begin{enumerate}[leftmargin=*]

\item \textbf{Policy Compilation (on policy admin side)}:  
   Each regex pattern appearing in the XACML policy (e.g., \texttt{J.* K.* Hibbert}) is compiled into a deterministic finite automaton (DFA) using the \texttt{regex\_automata} library. The DFA is then serialized into a compact byte array.

\item \textbf{Program Embedding}:  
   The serialized DFA byte sequence is embedded as a constant artifact within the zkVM program (e.g., linked as a static data segment). No regex parsing or automaton construction is needed during execution.

\item \textbf{Proof Generation (in zkVM)}:  
   During execution, the zkVM simply reconstructs the DFA from the serialized bytes and performs table-driven transitions over the input string. Pattern parsing and NFA/DFA construction are eliminated.
\end{enumerate}

\begin{listing}[!ht]
\begin{minted}[fontsize=\small, frame=none]{rust}
use regex_automata::{
    dfa::dense::DFA,
    nfa::thompson,
};
pub fn create_dfa_bytes(pattern: &str) -> 
    Result<Vec<u8>, Box<dyn std::error::Error>> {
    // 1. Build Thompson NFA
    let nfa = thompson::NFA::compiler()
                              .build(pattern)?;
    // 2. Determinize to a dense DFA
    let dfa = DFA::builder()
                    .build_from_nfa(&nfa)?;
    // 3. Serialize DFA to a little-endian byte
    #[cfg(target_endian = "little")]
    let (bytes, pad) = 
        dfa.to_bytes_little_endian();
    assert_eq!(pad, 0);
    Ok(bytes.to_vec())
}
\end{minted}
\caption{DFA precompilation using \texttt{regex\_automata}}
\label{lst:dfa_compile}
\end{listing}
\paragraph{Implementation of DFA precompilation.} We use the \texttt{regex\_automata} library's Thompson compiler and dense DFA builder to generate precompiled transition tables. The core Rust code is shown in Listing~\ref{lst:dfa_compile}.
During policy deployment, each regex pattern is passed to \texttt{create\_dfa\_bytes}, and the resulting byte vector is stored as a binary artifact (e.g., a \texttt{.dfa} file). For the pattern \texttt{J.* K.* Hibbert}, the deployed policy includes a serialized DFA encoding equivalent semantics.

\paragraph{DFA reconstruction and evaluation in zkVM. } Inside the zkVM, the DFA is reconstructed from its serialized form and used for deterministic, table-driven evaluation. Listing~\ref{lst:dfa_eval} illustrates this process. Compared to the naive implementation in Listing~\ref{lst:regex_naive}, this design achieves:

\begin{listing}[!ht]
\begin{minted}[fontsize=\small, frame=none]{rust}
use regex_automata::dfa::dense::DFA;
// Statically embedded DFA artifact
static DFA_BYTES: &[u8] = 
              include_bytes!("J_K_Hibbert.dfa");
pub fn eval_name_attr(input: &str) -> bool {
    // Reconstruct DFA from bytes, no parsing
    let dfa = DFA::from_bytes(DFA_BYTES)
                    .unwrap();
    // Table-driven pattern matching
    return dfa.try_search_fwd(&input);
}
\end{minted}
\caption{Table-driven DFA evaluation inside zkVM}
\label{lst:dfa_eval}
\end{listing}

\begin{itemize}[leftmargin=*]
    \item \textbf{No runtime parsing or construction overhead}:  
   All regex parsing, AST construction, and DFA determinization occur in the trusted pre-deployment phase. The zkVM only executes deterministic transitions.
   \item \textbf{Reduced trace and constraint complexity}:  
   Matching becomes a fixed table lookup operation. The number of RISC-V instructions and committed memory reads scales linearly with input length and is decoupled from the complexity of the regular expression, significantly reducing the overhead.
   \item \textbf{Clear trust and correctness boundaries}:  
   Since the policies are compiled by a trusted policy authority, we assume semantic equivalence between the regular expression in the original XACML and its compiled DFA. The zkVM only needs to prove correct evaluation over a known DFA, not correctness of the compilation itself.
\end{itemize}

This optimization reduces the number of executed RISC-V instructions and minimizes the volume of committed memory reads. As a result, it leads to fewer constraints and faster prover time, while retaining full correctness guarantees for regex-based attribute validation.



\subsubsection{JWT Parsing}\label{subsubsec:jwt}

Another common gadget in the policies is the JWT parsing. Listing~\ref{lst:jwt-example} shows an example of JWT. 

\begin{listing}[!ht]
\begin{minted}[fontsize=\small, frame=none]{python}
header = {
  "alg": "RS256",
  "typ": "JWT",
  "kid": key_id,
}
payload = {
  "iss": "https://login.example.com/",
  "subject_id": "Julius Hibbert",
  "aud": "api://payments-service",
  "exp": "<current_time> + 3600",
  "iat": "<current_time>",
  "auth_time": "<current_time> - 100",
  "email": "user@example.com",
  "nonce": "3e4f0f67-bc5a-413d-b528-93fd1c71fd4e",
  "roles": "admin",
}
\end{minted}
\caption{Example JWT used in \name.}
\label{lst:jwt-example}
\end{listing}


Unlike prior works, such as zkLogin~\cite{baldimtsi2024zklogin}, that develop special protocols verifying the JWT parsing with auxiliary input provided by the prover, we directly use the \texttt{serde\_json} library in Rust. This is because we observe that implementing the special protocol in zkVM incurs a similar overhead to using the \texttt{serde\_json} library, for common JWT tokens with no more than 10 fields in the payload. Moreover, directly using the \texttt{serde\_json} library avoids the vulnerability of the prover lying about the structure of the JWT token by injecting special characters in some fields of the payload.

\subsubsection{RSA Verification}\label{subsubsec:rsa}
Finally, another major overhead comes from the signature verification of the JWT. The RSA verification involves the SHA-2 hashing of the \texttt{header || '.' || payload} string and modular exponentiations over the large RSA group. Naively implementing them in zkVM would result in the prover time of more than 1,000 seconds. 

In \name, we utilize the pre-compiled system calls in the RISC Zero zkVM, including SHA-256 and big-integer modular arithmetic (\texttt{sha2} and \texttt{rsa} in~\cite{risc0_precompile}), to improve the efficiency. These pre-compiled modules are verified and committed implementations of cryptographic primitives that are included in the zkVM's instruction set. Their correctness can be verified publicly by the PA, and thus using these pre-compiled instructions do not alter the soundness of the proof.  
Additionally, we avoid the dynamic construction of RSA public key parameters. Instead of decoding the base64-encoded modulus and exponent in zkVM, the public key is hardcoded in the policy as a static field in the program's read-only memory. This removes expensive operations and big-integer initializations. 

As we will show in the experiments, although this optimization improves the prover time of RSA verification by one to two orders of magnitude, its overhead is still quite high in practice and is the bottleneck of the prover. In the literature, there are optimized circuits/R1CS constraints~\cite{kosba2018xjsnark,circom_rsa} and special protocols~\cite{woo2025efficient} for RSA verification. It merely takes around 2 seconds to generate a proof using Groth16~\cite{Groth2016}, and the prover time would be another order of magnitude faster using the SNARK backend of RISC Zero. In principle, we could utilize them through a pre-compiled instruction, or through a proof composition using recursive SNARKs to combine the proof for RSA verification with the zkVM proof. Unfortunately, we were not able to integrate them to RISC Zero. We believe this is only a gap of the implementation, but not any fundamental issue of our scheme.



\section{Experiments}\label{sec:exp}



\begin{figure*}[t!]
    \centering
    \resizebox{0.95\linewidth}{!}{\input{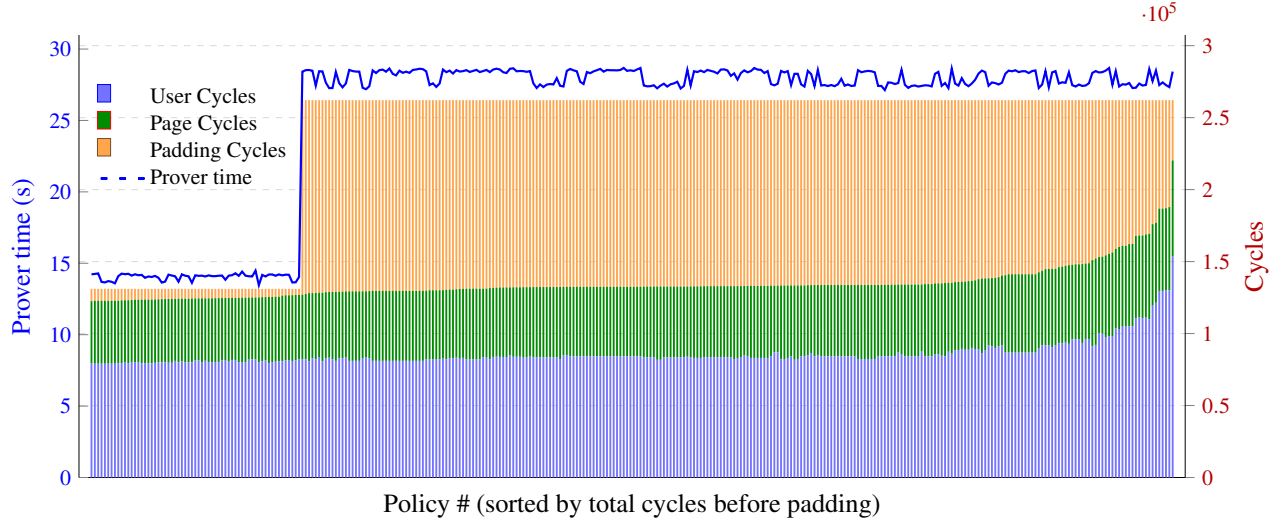}}
    \caption{Prover time and cycles of all policies.}
    \label{fig:end2end-cases-dual-sorted-user}
\end{figure*}

\begin{figure*}[t!]
    \centering
    \resizebox{0.95\linewidth}{!}{\input{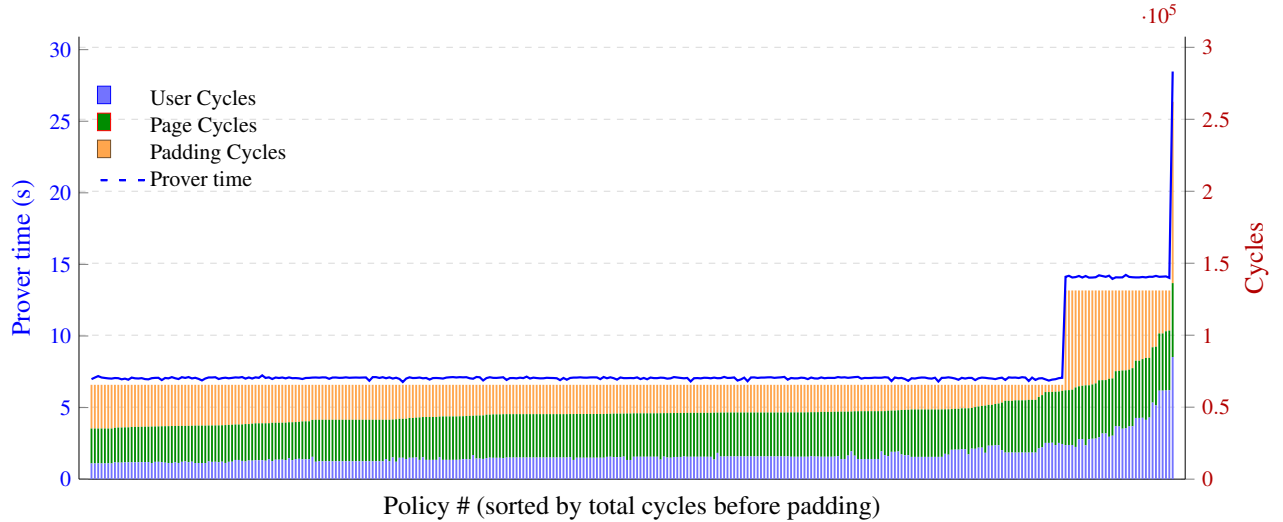}}
    \caption{Prover time and cycles of all policies. (without RSA verification)}
    \label{fig:end2end-cases-dual-sorted-user-norsa}
\end{figure*}

We have fully implemented \name, and we report the empirical evaluations in this section. 

\paragraph{Settings.} We tested the prover and the verifier on an AWS c7i.4xlarge instance with 16 vCPUs of Intel Xeon Scalable and 32 GB of memory. The prover and the verifier are both implemented in Rust and compiled with \texttt{cargo} using the \texttt{--release} flag to enable optimizations. We use the RISC Zero zkVM v3.0.3~\cite{risc0} to instantiate our SNARK. It has a transparent setup, and the security relies on collision-resistant hash functions and the Fiat-Shamir transformation.


\subsection{End-to-end Performance}

\paragraph{Dataset, compilation and verification key generation.} As presented in Section~\ref{subsubsec:compiler_exp}, we test \name on the XACML 3.0 conformance test suite~\cite{authzforce2025}, and our compiler supports 323 out of 389 testable policies. On average, it takes around 0.24 seconds to compile a policy from XACML to Rust, while it takes 1.73s to compile the policy from Rust to RISC-V and also generate the verification key of the policy in zkVM. This part is done only once for each policy by the PA, and thus the overhead is very small in practice.

\paragraph{Prover time.} The prover time of every policy \name supports is presented in Figure~\ref{fig:end2end-cases-dual-sorted-user}. As shown in the figure, the prover time is either around 14 seconds or around 28 seconds. This is due to the padding of RISC Zero zkVM. To provide a better understanding, we also plot the number of \emph{total cycles} before padding for each policy reported by RISC Zero, and the figure is sorted by total cycles. The total cycles consist of \emph{user cycles} (the number of cycles executed by the main Rust code),
and additional instructions inserted by RISC Zero that are required for the zkVM. The total number of cycles is then padded to the nearest power of 2. Therefore, for those policies with fewer than 130K total cycles before padding, they are padded to $2^{17} = 131,072$ cycles, while others are padded to $2^{18} = 262,144$ cycles. The same number of cycles after padding does not result in exactly the same prover time because of different instructions used in different policies, but they are highly correlated. We also tested an empty Rust code in RISC Zero, and it results in $2^{15}=32,768$ cycles after padding with the prover time of 3.8 seconds. There are around 20K paging cycles and 9K reserved cycles in the padding for initializing the memory and setting up the zkVM.

To demonstrate the potential of \name, we remove the RSA verification module from the zkVM and report the prover time in Figure~\ref{fig:end2end-cases-dual-sorted-user-norsa}. This is because as explained in~\Cref{subsubsec:rsa}, we envision that optimized circuits or special protocols for RSA verification should be integrated in the future, eliminating the overhead of this module in zkVM. As shown in the figure, the prover time for 88.6\% of the policies drops to only 7 seconds. It shows that the main bottleneck of the end-to-end prover time in \name comes from the RSA verification and the padding of RISC Zero. In fact, the number of user cycles for the majority of policies is less than 20K, which is even smaller than the additional padded cycles of RISC Zero for initialiation. With the recent improvement of zkVMs~\cite{ethproofs} with hardware accelerations, our scheme can be made practical in the near future. We further investigated the cases with more than 100K user cycles in Appendix~\ref{app:exp}.


\paragraph{Proof size and verifier time.} The proof size of all policies are between 238KB to 250KB, and the verifier time are between 14 and 15.5 milliseconds. The verifier time is very practical for edge devices in practice. The proof size can be further compressed to 256 bytes via recursive SNARKs (e.g., Groth16~\cite{Groth2016}), with an additional overhead on the prover time. 

\subsection{Ablation Study and Micro-benchmarks}

To better understand the improvement of \name over naive approaches, we present an ablation study and micro-benchmarks in this section. 

We take the policy with the largest number of user cycles, implement it in RISC Zero naively without any optimizations, and then add each of our optimizations one by one cumulatively. Table~\ref{tab:ablation} shows the ablation study. It would take 1,301.7 seconds to generate a proof in the naive approach without any optimizations, while it only takes 27.3 seconds in \name, which is 47.7$\times$ faster. The largest improvement comes from the RSA verification, which can be further improved by a dedicated circuit/R1CS outside zkVM, or by special protocols such as~\cite{woo2025efficient}. Moreover, the regex optimization targets a major source of user cycles, significantly reducing user cycles and lowering prover time from 176s to 27.3s.

\begin{table}[t]
    \centering
    \caption{Ablation study of Prezta optimizations. Each row denotes adding the optimization in Column 1 cumulatively. }
    \label{tab:ablation}
    {\small
    \begin{tabular}{lcccc}
        \toprule
        Optimization & User & Prover & Verifier& Proof\\
         & cycles & time (s) &time (ms) & size (KB) \\
        \midrule
        None & 10{,}521{,}257 & 1301.7 & 196.9 & 3236 \\
        RSA & 805{,}388 & 176 & 32.8 & 536 \\
        Regex & 153{,}677 & 27.3 & 15.1 & 250 \\
        \bottomrule
    \end{tabular}
    }
\end{table}
\begin{table}[t]
    \centering
    \caption{Breakdown of user cycles by major functions.}
    \label{tab:user-cycles-breakdown}
    {\small
    \begin{tabular}{lcccc}
        \toprule
         & Others & Regex & RSA \\
        \midrule
        User cycles & 13{,}508 & 71{,}429 & 68{,}598 \\
        \bottomrule
    \end{tabular}
    }
\end{table}

In Table~\ref{tab:user-cycles-breakdown}, we also provide a breakdown of \name for the same policy in terms of user cycles. As shown in the table, the RSA verification contributes to 45\% of the user cycles and our optimized Regex contributes to 47\%\footnote{The sum of cycles does not equal to the last row of Table~\ref{tab:ablation} because RISC Zero also inserts different number of user cycles when we implement each module individually.}.

\paragraph{Comparison to Reef~\cite{angel2024reef}.} In Appendix~\ref{app:com_reef}, we further compare the performance of our regex-focused optimizations with Reef~\cite{angel2024reef}, a SNARK system tailored for regular expressions.

\section{Discussion}

We now address the practical considerations that arise when deploying \name, including update mechanics, proof generation models, and auditability.

\subsection{Updates}
\label{sec:updates}
One important feature of the architecture is the ease of updates. To update a policy, the policy authority needs to compile it from XACML to Rust and then use RISC Zero to obtain a new \texttt{IMAGE\_ID}. This ID is then signed and distributed to edge devices (perhaps piggybacked on the requests). To ensure the most recent policy is used, the signed \texttt{IMAGE\_ID}s should be accompanied by an expiration time; more speedy revocation could be supported by using a protocol similar to OCSP~\cite{rfc6960} to certify that a policy is still in effect. Specifically, a revocation service tracks whether a policy has been revoked and, when queried by a client, signs a status message indicating that status. The client includes this status message as another input to the policy evaluation program, much like the JWT from the IdP. The program then checks the signature, revocation status, and freshness against the current timestamp (\Cref{tab:input_attributes}). The edge device remains agnostic to revocation checks. As with OCSP, the policy authority may operate the revocation service directly or delegate it to another server trusted only to track status. Because these checks are embedded in the PA-signed policy, \name can vary parameters such as freshness requirements, potentially conditioned on other attributes such as request sensitivity.

We note that some of the logic of verifying update signatures, expirations, and revocation status, could itself be encoded into a separate (long-lived) zkSNARK or zkVM program to allow updates without any specific code on the edge device. 

Due to the flexibility of the zkVM, more significant changes are possible, all without device upgrades. For example, to support a new policy language, such as Rego~\cite{rego_language} or NGAC~\cite{ngac}, it would be necessary to build a new compiler that translates such policies into Rust, but from the point of view of the edge device, the new programs would be equivalent to a policy update. One could also switch identity management systems, e.g., from JWT to SAML~\cite{oasis_saml_2005}, or from RSA signatures on tokens used in our prototype to ECDSA or others. Again, since the edge device does not interact with authentication tokens, it would remain agnostic to this change. (Unless new precompiles were necessary to efficiently verify the signature, in which case the verifier \emph{would} need to be updated.)

In fact, it would be possible to create a post-quantum secure version of \name.  The proof algorithm used by RISC Zero~\cite{bensasson_et_al:LIPIcs.ITCS.2020.5} (see~\cite{Bruestle2023RISCZero} for details) is hash-based and thus not vulnerable to quantum attacks. To render the entire system post-quantum secure, it would be necessary to use a post-quantum signature algorithm for the signing of policy \texttt{IMAGE\_ID}s and the credentials (JWTs) in the identity system. The latter signature must be verified inside the zkVM, so a version of SPHINCS+~\cite{sphincsplus2019} instantiated with Poseidon~\cite{grassi2021poseidon} would be a good candidate. 

Verifier upgrades would be needed for any updates to the zkVM verifier, and these would become critical if these addressed a soundness error in the proof system. The field of zkVMs is rapidly maturing and we expect that a robust, production-ready zkVM will be available in the near future.

\subsection{Proof Outsourcing and Composition}

The resources needed to create a proof may be prohibitive for low-power devices, such as phones or monitoring terminals. In such cases, proof generation can be outsourced to either a local server or a cloud service; this is a strategy adopted by, e.g., zkLogin~\cite{baldimtsi2024zklogin}. Importantly, this service would not be a trusted component---its compromise would at worst obtain a log of attempted accesses, but would not enable unauthorized access, due to the soundness of zkSNARKs.\footnote{%
One has to be somewhat careful with supplying JWTs to an outsourced prover, as they serve as bearer certificates and could therefore be abused by an outsourced prover. Again, following a strategy from zkLogin, we can bind the JWT by having it including a hash of the request in the nonce, which could then be verified by the zkSNARK.}%
Full access privacy could be achieved by using a committee of servers using recent collaborative zkSNARKs~\cite{ozdemir2022experimenting,garg2023zksaas}. A powerful cloud proving server with GPUs could also significantly reduce the proving latency. 

RISC Zero supports \emph{proof composition}, where one zkVM program can verify another’s execution. This could be used to support modular policy programs, with parts of the policy delegated to other programs (or even other administrators), and would create potential reuse of proof components to optimize performance. We leave a full exploration of proof composition in this context to future work.

\subsection{Auditing}
\label{sec:audit}

To support audits, policies could require subjects to register a log record at a logging server and supply a proof of inclusion with the access request, which would be verified by the policy program. The log server could be made transparent and append-only by using techniques from certiicate transparency~\cite{Laurie2013CertificateTransparency}; the log entries could also be encrypted using techniques similar to Larch~\cite{dauterman2023larch}.


\section*{Acknowledgments}

This material is based upon work supported by
the National Science Foundation (NSF) under Grant No. 2113819 and No. 2401481. Any opinions, findings, and conclusions or recommendations expressed in this material are those of the author(s) and do not necessarily reflect the views of these institutes.

\section*{Ethical Considerations}

Prezta is a proposed architecture for OT security, and this paper is a first step toward possible future production systems based on that architecture. Any eventual deployment could affect OT operators, administrators, other parties interacting with edge devices, and indirectly the public who rely on critical infrastructure.

We believe Prezta has the potential to improve security and reduce management burden in OT environments. However, deployment of new technology in such environments also introduces risks, including implementation flaws, misconfiguration, new attack surfaces, and possible availability impacts. These tradeoffs are highly deployment-specific, so a full cost-benefit-risk analysis is beyond the scope of this paper.

We judge the risks from publication of this work and code to be minimal. The work proposes a defensive architecture and our implementation is a research prototype rather than a production-ready system. For clarity, we will explicitly label the implementation as unaudited research code.

Like many defensive security technologies, a system that strengthens OT authorization could also make some forms of access or exploitation more difficult for a range of actors, including sophisticated state-backed actors. We nevertheless view improved security and resilience for OT and critical infrastructure as a net benefit, given the broad set of stakeholders who depend on these systems.

\section*{Open Science}
The artifact is permanently archived at \url{https://doi.org/10.5281/zenodo.20303462} (CC BY 4.0). It includes: (1)~the Prezta prototype (RISC Zero zkVM v3.0.3) and (2)~the XACML-to-Rust policy compiler. Source repositories: \url{https://github.com/walotta/ZK_Zero_Trust} and \url{https://github.com/osaidameer/xacml-to-rust}. The implementation is unaudited research code.






\bibliographystyle{abbrv} 
\bibliography{zkp,zkauth,more_refs}

\appendix

\section{Additional Preliminaries}\label{app:prelim}

\subsection{SNARK}

\begin{definition}[Succinct non-interactive argument of knowledge (SNARK)]\label{def:snark} Let $\mathcal{R}$ be a relation with public instance $x$ and private witnesses $w$, a SNARK for $\mathcal{R}$ consists of PPT algorithms ($\mathcal{G}, \mathcal{P}, \mathcal{V}$) with the following properties:

\begin{itemize}
    \item Completeness. For any instance $(x, w)\in \mathcal{R}$

\begin{equation*}    
\Pr \left[ 
\begin{tabular}{ c|c }
 \multirow{2}{6.5em}{$\mathcal{V}(x, \vk, \pi)=1$} & $(\pk, \vk) \gets \mathcal{G}(1^\lambda)$, \\
  & $\pi \gets \mathcal{P}(x, w, \pk)$, \\
\end{tabular}
\right]=1 
\end{equation*}

    \item Knowledge Soundness: for any PPT adversary $\mathcal{A}$, there exists an expected PPT knowledge extractor $\mathcal{E}_\mathcal{A}$ such that the following probability is $\le \mathsf{negl}(\lambda)$:

\begin{equation*}    
\Pr \left[ 
\begin{tabular}{ c|c }
 $\mathcal{V}(x, \vk, \pi^*)=1$
  & $(\pk, \vk) \gets \mathcal{G}(1^\lambda)$, \\
  $ \land (x, w)\notin \mathcal{R}$ & $(\pi^*;w) \gets (\mathcal{A}||\mathcal{E}_\mathcal{A})(x, \pk)$ \\
\end{tabular}
\right].
\end{equation*}

 
\item Succinctness. The proof size $|\pi|$ is sublinear in the size of the relation $\mathcal{R}$ and the witness $|w|$.

\end{itemize}

\end{definition}

In addition to these properties, a zero-knowledge SNARK further satisfies zero-knowledge, which informally says that the proof leaks no information about the witness beyond the fact that the instance is in $\mathcal{R}$. We only use the soundness and the succinctness of a SNARK in this paper and omit the formal definition of zero-knowledge.  


\section{Proof Sketch of Theorem~\ref{thm:main}} \label{app:proof}

The completeness is implied by the correctness of the digital signature scheme and the completeness of the SNARK on the relation $\mathcal{R}$ checking $P(\mathsf{att},Req,\info)= 1$ and $\mathsf{signature}.\verify(\mathsf{att},\sigma_\mathsf{att}, \pk_\mathsf{IdP})=1$.

For security, when $\zkzt.\verify(\pi, Req, \info, \vk_P, \sigma_P,$ $\pk_\mathsf{PA})=1$, by the knowledge soundness of SNARK (\Cref{def:snark}), there exists an extractor that can extract $w = (\mathsf{att}, \sigma_\mathsf{att})$ such that $P(\mathsf{att},Req,\info) \neq 1$ or $\mathsf{signature}.\verify(\mathsf{att},\sigma_\mathsf{att}, \pk_\mathsf{IdP}) \neq 1$ only with negligible probability. Then by the unforgeability of the digital signature, when $\mathsf{signature}.\verify(\mathsf{att},\sigma_\mathsf{att}, \pk_\mathsf{IdP}) = 1$, the probability that $\mathsf{att}\notin Q_\mathsf{att}$ is negligible. Moreover, by Step 2 of $\zkzt.\verify$, $\mathsf{signature}.\verify(\vk_P,\sigma_P, \pk_\mathsf{PA}) = 1$, the probability that $\vk_P \notin Q_P$ is negligible. The by the union bound, the probability that $\zkzt.\verify(\pi, Req, \info, \vk_P, \sigma_P, \pk_\mathsf{PA})=1$ and $P(\mathsf{att},Req,\info)\neq 1$, or $\verify(\mathsf{att}, \sigma_\mathsf{att},\pk_\mathsf{IdP})= 1$ and $\mathsf{att}\notin Q_\mathsf{att}$, or $\verify(\vk_P, \sigma_P,\pk_\mathsf{PA})= 1$ and $\vk_P\notin Q_P$ is negligible, completing the proof of security.

\section{Unsupported XACML Functions}\label{app:unsupported}
\begin{table}[h!]
\centering
\caption{Unsupported XACML Elements, Functions and Features}
{\small
\begin{tabular}{|l|l|}
\hline
\textbf{Name} & \textbf{Type} \\
\hline
x500-* & Functions \\
base64Binary-* & Functions \\
hexBinary-* & Functions \\
rfc822Name-* & Functions \\
dnsName-* & Functions \\
ipAddress-* & Functions \\
*-add-yearMonthDuration & Functions \\
*-subtract-yearMonthDuration & Functions \\
*-add-dayTimeDuration & Functions \\
*-subtract-dayTimeDuration & Functions \\
double-*\{set functions\} & Functions \\
any-of & Function \\
all-of & Function \\
any-of-any & Function \\
any-of-all & Function \\
all-of-any & Function \\
all-of-all & Function \\
map & Function \\
AttributeSelector & Element \\
PolicyIdReference & Element \\
PolicysetIdReference & Element \\
Obligations & Feature \\
Advice & Feature \\
\hline
\end{tabular}
}
\end{table}

\section{Full Policy vs Generated Code}\label{app:codegen}

\begin{listing}[!ht]
\begin{minted}[fontsize=\small, frame=none, breaklines, breakanywhere]{xml}
<Policy PolicyId="policy" RuleCombiningAlgId="deny-overrides">
 <Target/>
  <Rule Effect="Permit" RuleId="rule">
   <Target>
    <AnyOf>
     <AllOf>
      <Match MatchId="string-equal">
       <AttributeValue DataType="#string">Julius Hibbert</AttributeValue>
        <AttributeDesignator AttributeId="subject-id" Category="access-subject" DataType="#string"/>
      </Match>
     </AllOf>
    </AnyOf>
    <AnyOf>
     <AllOf>
      <Match MatchId="anyURI-equal">
       <AttributeValue DataType="#anyURI">http://medico.com/record/patient/BartSimpson</AttributeValue>
       <AttributeDesignator AttributeId="resource-id" Category="resource" DataType="#anyURI"/>
      </Match>
     </AllOf>
    </AnyOf>
    <AnyOf>
     <AllOf>
      <Match MatchId="string-equal">
       <AttributeValue DataType="#string">read</AttributeValue>
       <AttributeDesignator AttributeId="action-id" Category="action" DataType="#string"/>
      </Match>
     </AllOf>
     <AllOf>
      <Match MatchId="string-equal">
       <AttributeValue DataType="#string">write</AttributeValue>
       <AttributeDesignator AttributeId="action-id" Category="action" DataType="#string"/>
      </Match>
     </AllOf>
    </AnyOf>
  </Target>
 </Rule>
</Policy>
\end{minted}
\caption{Simplified Sample XACML Policy}
\label{lst:full-xacml-sample}
\end{listing}

\begin{listing}[!ht]
\begin{minted}[fontsize=\small, frame=none, breaklines, breakanywhere]{rust}
// imports
// jwt verification functions 
#[derive(Debug, PartialEq)]
enum Result {
    Permit,
    Deny,
    NotApplicable,
}
fn evaluate_target_rule(inp: &Inputs) -> bool {
    (("Julius Hibbert" == inp.access_subject_subject_id) && ("http://medico.com/record/patient/BartSimpson" == inp.resource_resource_id) && (("read" == inp.action_action_id) || ("write" == inp.action_action_id)))
}
fn evaluate_rule(inp: &Inputs) -> Result {
    if !evaluate_target_rule(inp) {
        return Result::NotApplicable;
    }
    return Result::Permit;
}
fn evaluate_target_policy(inp: &Inputs) -> bool {
    true // empty target, match all
}
fn evaluate_policy(inp: &Inputs) -> Result {
    if !evaluate_target_policy(inp) {
        return Result::NotApplicable;
    }
    let results = vec![evaluate_rule(inp)];
    //deny-overrides
    let mut atleast_one_permit = false;
    for res in &results {
        if *res == Result::Deny {
            return Result::Deny;
        } else if *res == Result::Permit {
            atleast_one_permit = true;
        }
    }
    if atleast_one_permit {
        return Result::Permit;
    }
    return Result::NotApplicable;
}
fn main() {
    let inp: Inputs = env::read();
    let mut decision = match evaluate_policy(&inp) {
        Result::Permit => true,
        _ => false,
    };
    let jwt: String = env::read();
    let jwt_positions: Vec<usize> = env::read();
    if !extract_jwt(&jwt, &jwt_positions, &inp) {
        decision = false;
    }
    env::commit(&decision);
    env::commit(&inp);
}
\end{minted}
\caption{Rust code generated from policy in Listing \ref{lst:full-xacml-sample}}
\label{lst:full-rust-generated}
\end{listing}


\section{Additional Experimental Results}\label{app:exp}

\paragraph{Policies with large number of user cycles.} We further investigated the cases with more than 100K user cycles in Figure~\ref{fig:end2end-cases-dual-sorted-user}. 
The large number of user cycles are primarily due to their intensive computational patterns during the policy evaluation. Specifically, cases \textbf{IIC150--IIC156} and \textbf{IIC340--IIC349} repeatedly invoke \texttt{IsoDuration::parse} or \texttt{DateTime} parsing functions for each attribute or bag element, resulting in frequent ISO-8601 duration and timestamp conversions. Cases \textbf{IIC201--IIC205} additionally construct and manipulate \texttt{HashSet} objects to perform intersection, union, subset, and equality operations, which significantly increase the number of cycles. The most computationally demanding cases, \textbf{IIC056} and \textbf{IIC057}, involve loading and executing dual DFA-based regular expressions in conjunction with JWT field validation, producing extensive state transitions and branching. Overall, these patterns combine repeated parsing, string processing, and complex set operations, leading to a substantially higher number of user cycles and thus characterizing these cases as representative high-load scenarios in the policy evaluation workload.

\subsection{Comparison with Reef}\label{app:com_reef}

To quantify the improvement of our regex-focused optimizations, we evaluate the six representative regular expressions in the policies from the dataset and compare them with Reef~\cite{angel2024reef}, a SNARK system tailored for regular expressions. Table~\ref{tab:regex-policies-mini} (top) reports the prover time for each regular expression. As shown in the table, even for short strings and simple regular expressions, the naive Regex library in zkVM leads to a significant overhead on the prover time. Our optimization improves the efficiency by one to two orders of magnitude. Compared to Reef, our protocol is even slightly faster for some of the cases, while slightly slower for others. Due to the limited number of examples from the XACML policy datasets, we also create synthetic regex expressions that may occur in authorization policies, such as checks on email, IP address and JWT. See the full expressions in Table~\ref{tab:full_regex}. The performance is reported in Table~\ref{tab:regex-policies-mini} (bottom) as well, and we observe similar comparisons with the naive approach and Reef.


The full regular expressions and inputs corresponding to the names used in Table~\ref{tab:regex-policies-mini} are shown in Table~\ref{tab:full_regex}.

\begin{table}[t]
    \centering
    \caption{Per-policy prover time (in seconds) for regex workloads. Short names are used in the table; the full regex patterns and inputs are listed in Table~\ref{tab:full_regex}.}
    \label{tab:regex-policies-mini}
    
    {\small
    \begin{tabular}{lccc}

        \toprule
        Short & RISC0 (Regex lib) & Prezta & Reef \\
        \midrule
        \texttt{RW\_READ} & 61.5 & 7.2 & 9.0 \\
        \texttt{RW\_DELETE} & 61.6 & 7.3 & 10.8 \\
        \texttt{J\_HIBBERT} & 314.4 & 14.8 & 9.4 \\
        \texttt{B\_SIMPSON} & 509.0 & 14.9 & 11.3 \\
        \texttt{J\_K\_HIBBERT} & 127.3 & 14.7 & 11.6 \\
        \texttt{B\_O\_SIMPSON} & 128.2 & 14.8 & 11.6 \\
        \midrule\midrule
        \texttt{EMAIL\_SIMPLE} & 315.0 & 7.4 & 11.3 \\
        \texttt{EMAIL\_COMPLEX} & 127.3 & 7.5 & 12.3 \\
        \texttt{IPV4} & 126.8 & 14.9 & 13.4 \\
        \texttt{IPV6} & 61.3 & 7.5 & 13.5 \\
        \texttt{JWT\_LIKE} & 129.7 & 7.4 & 11.4 \\
        \texttt{OAUTH\_PATH} & 54.1 & 7.4 & 11.1 \\
        \texttt{ISO8601} & 111.8 & 7.4 & 12.4 \\
        \bottomrule
    
    \end{tabular}
    }
\end{table}

\begin{table*}[h]
    \centering
    \small
    {\small
    \begin{tabular}{@{}lp{0.75\linewidth}@{}}
        \toprule
        \texttt{RW\_READ} & pattern: \lstinline!read|write!; input: \lstinline!read! \\
        \texttt{RW\_DELETE} & pattern: \lstinline!read|write!; input: \lstinline!delete! \\
        \texttt{J\_HIBBERT} & pattern: \lstinline!J.* Hibbert!; input: \lstinline!Julius Hibbert! \\
        \texttt{B\_SIMPSON} & pattern: \lstinline!B.* Simpson!; input: \lstinline!Julius Hibbert! \\
        \texttt{J\_K\_HIBBERT} & pattern: \lstinline!J.* K.* Hibbert!; input: \lstinline!Julius Hibbert! \\
        \texttt{B\_O\_SIMPSON} & pattern: \lstinline!B.* O.* Simpson!; input: \lstinline!Julius Hibbert! \\
        \texttt{EMAIL\_SIMPLE} & pattern: \lstinline!^[a-zA-Z0-9._%+-]+@[a-zA-Z0-9.-]+!\\
        & \lstinline!.[a-zA-Z]{2,}$!; input: \lstinline!user@example.com! \\
        \texttt{EMAIL\_COMPLEX} & pattern: \lstinline!^[a-zA-Z0-9._%+-]+@[a-zA-Z0-9.-]+!\\
        & \lstinline!\\.[a-zA-Z]{2,}$!; input: \lstinline!very.unusual.@.unusual-example@domains.example.com! \\
        \texttt{IPV4} & pattern: \lstinline!^((25[0-5]|2[0-4][0-9]|[01]?[0-9][0-9]?).){3}!\\
        & \lstinline!(25[0-5]|2[0-4][0-9]|[01]?[0-9][0-9]?)$!; input: \lstinline!192.168.1.1! \\
        \texttt{IPV6} & pattern: \lstinline!^(?:[0-9a-fA-F]{1,4}:){7}[0-9a-fA-F]{1,4}$!\\
        & \lstinline!|^::1$|^::$!; input: \lstinline!2001::1! \\
        \texttt{JWT\_LIKE} & pattern: \lstinline!^ey[A-Za-z0-9+/=]+\.[ey][A-Za-z0-9+/=]+!\\
        & \lstinline!\.[A-Za-z0-9+/=]*$!; input: \lstinline!not.a.jwt.token! \\
        \texttt{OAUTH\_PATH} & pattern: \lstinline!^/oauth/(authorize|token|revoke)$!; input: \lstinline!/oauth/authorize! \\
        \texttt{ISO8601} & pattern: \lstinline!^([0-9]{4})-([0-9]{2})-([0-9]{2})T([0-9]{2}):([0-9]{2})!\\
        & \lstinline!:([0-9]{2})(.[0-9]+)?(Z|([+-][0-9]{2}):([0-9]{2}))\$!; input: \lstinline!2025-04-05T10:15:30.123+05:30! \\
        \bottomrule
    \end{tabular}
    }
    \caption{Full regex patterns and inputs.}\label{tab:full_regex}
\end{table*}

\end{document}